\documentclass{article}

\PassOptionsToPackage{numbers,sort,compress}{natbib}
\usepackage[main, final]{neurips_2026}
\usepackage{microtype}
\usepackage{graphicx}
\usepackage{subcaption}
\usepackage{booktabs} 

\usepackage{hyperref}
\usepackage{float}
\usepackage[most]{tcolorbox}
\usepackage{wrapfig}
\usepackage{booktabs}
\usepackage{tabularx}
\usepackage{enumitem}

\usepackage{amsmath}
\usepackage{amssymb}
\usepackage{mathtools}
\usepackage{amsthm}

\usepackage[capitalize,noabbrev]{cleveref}

\theoremstyle{plain}
\newtheorem{theorem}{Theorem}
\newtheorem{proposition}[theorem]{Proposition}
\newtheorem{lemma}[theorem]{Lemma}

\theoremstyle{definition}
\newtheorem{definition}{Definition}
\newtheorem{assumption}{Assumption}
\theoremstyle{remark}
\newtheorem{remark}{Remark}

\usepackage{tikz}
\usetikzlibrary{arrows.meta,positioning,fit,shapes.multipart}

\newcommand{\E}{\ensuremath{\mathbb{E}}} 

\newcommand{\lin}{\mathrm{lin}}

\newcommand{\indep}{\perp\!\!\!\!\perp}

\usepackage[textsize=tiny]{todonotes}

\title{The Subjectivity of Monoculture}

\author{%
  Nathanael Jo \\
  MIT EECS \\
  \texttt{nathanjo@mit.edu} \\
  \And
  Nikhil Garg \\
  Cornell Tech \\
  \texttt{ngarg@cornell.edu} \\
  \And
  Manish Raghavan \\
  MIT Sloan \& EECS\\
  \texttt{mragh@mit.edu} \\
}

\begin{document}

\maketitle

\begin{abstract}
Machine learning models -- including large language models (LLMs) -- are often said to exhibit \textit{monoculture}, where outputs agree strikingly often.
Our work begins with a pair of experiments that challenge this narrative.
We find that monoculture estimates can virtually disappear based on
choices made by the analyst.

Informed by these findings, we theoretically formalize what it means to measure monoculture. We show that monoculture is inherently subjective, relying on two key decisions: \textit{1)} the baseline null model for what ``independence'' should look like; and \textit{2)} the population of models and items under consideration. We conclude with concrete guidelines for future evaluators to ensure the robustness of their results.

\end{abstract}

\section{Introduction}\label{sec:introduction}

A growing body of empirical work suggests that AI models, including both predictive and generative models, produce homogeneous outputs. Across domains---including text/image generation and decision making---separately trained models often produce strikingly similar responses \cite{wu2025generative, padmakumar2024does, raghavan2024competition, ugander2024art, doshi2024generative, goelgreat2025, kimcorrelated2025, toups2023ecosystem}. Throughout the literature, a common perspective has emerged: Models don't just agree, they agree \emph{too much}, resulting in ``algorithmic monoculture'' \cite{kleinberg2021algorithmic}. 


Why does monoculture matter? 
In resource allocation settings such as hiring or lending, similar predictive systems may lead institutions to systematically privilege or disadvantage the same groups \cite{jain2024algorithmic}. In market settings, firms that independently deploy similar pricing algorithms may inadvertently sustain higher prices and harm consumers \cite{jo2025homogeneous}. More broadly, diversity is often a desirable property of sociotechnical systems because it can improve robustness to error, reflect a plurality of norms or preferences, and promote creative or novel outputs. A large literature has emerged analyzing the downstream consequences of algorithmic monoculture
\cite{bommasani2022pickingpersondoesalgorithmic,creel2022algorithmic,peng2024monoculture, peng2024wisdomfoolishnessnoisymatching,baek2025hiring,ren2026hitting}.

However, claims about monoculture are by nature relative, because it suggests that outputs are more correlated than one would expect under some ideal world. As we will show, this comparison cannot be made in a vacuum. Instead, it requires two subjective choices made by the analyst, and these choices drive the conclusions they reach about monoculture:

\begin{enumerate} [itemsep=0.2em, topsep=0pt, parsep=0pt]
    \item What is the baseline (null model) against which agreement is measured?
    \item On what population of models and instances do we measure excess agreement?
\end{enumerate}

\vspace{-0.5em}

\paragraph{Baselines for excess agreement.}

The literature contains two distinct approaches to constructing baselines against which excess agreement is measured.
One is to use human subjects \cite{wu2025generative, padmakumar2024does, ugander2024art, doshi2024generative,anderson2024homogenization}. While human baselines are useful for normative comparisons about diversity, they require committing to a reference population and a notion of human diversity.

A second approach, and one we take here, is to measure excess agreement relative to a ``null model,'' under which responses are conditionally independent \cite{kimcorrelated2025, goreckimonoculture2025, goelgreat2025, toups2023ecosystem}. A null model is a particular data-generating process that creates an expectation of how much agreement is reasonable. Any observed agreement in excess of the null model is potentially evidence of monoculture. A popular null model is capability-based \cite{goelgreat2025, jo2025homogeneous}.  In its simplest form, two models with accuracies $(p_1,p_2)$ are viewed as independent if their expected agreement is the chance overlap $p_1 p_2 + (1-p_1)(1-p_2)$ on each item. 
\begin{wrapfigure}{r}{0.5\textwidth}
    \centering
    \includegraphics[width=\linewidth]{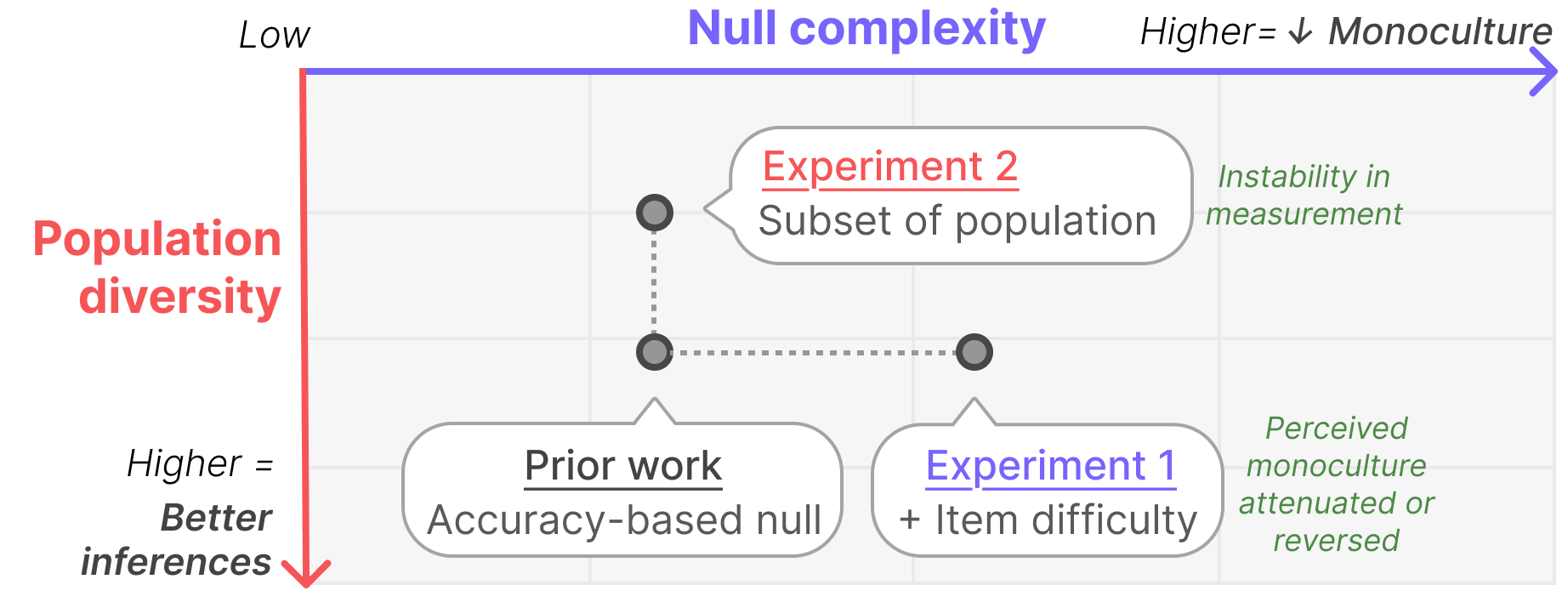}
    \caption{Summary of our main contributions.}
\label{fig:summary_null}
\end{wrapfigure}

Is this the ``right'' null model? Perhaps---but this baseline might ignore other
relevant features such as \textbf{question difficulty}: if two models
both succeed on easy items and fail on hard ones, their agreement may exceed the $p_1p_2 + (1-p_1)(1-p_2)$ baseline even in the absence of any meaningful
dependence. Or perhaps some models perform well on math questions, whereas others
excel at writing code. Again, this would lead to excess agreement beyond what
the simple baseline would predict (\textcolor{blue}{Null complexity} in Fig~\ref{fig:summary_null}).
Generally, a null model establishes the structure we \textit{expect} to see in data, giving us a baseline against which to measure
\textit{unexpected} correlation.

Complicating things further is that we generally lack ``ground truth''
information about latent structural properties like 
item difficulty. Instead, we must infer them from
data: e.g., the difficult questions are precisely the ones that many models get
wrong.
With a small set of models, one might incorrectly attribute all correlation
to heterogeneous item difficulty, thereby failing to detect any excess correlation.
In other words, latent structural properties like item difficulty or model
capability
cannot be defined in a context-free way; they are defined only relative to
a population of models and items.
As a result, even after fixing a null model, conclusions about excess agreement depend critically on the set of models and questions over which it is fit and evaluated (\textcolor{red}{Population diversity} in Fig~\ref{fig:summary_null}).

\textbf{Our contributions.} We begin in Section~\ref{sec:experiments} by
demonstrating empirically that inferences about monoculture can change drastically given different choices on \textcolor{blue}{null complexity} and \textcolor{red}{population diversity}. Motivated by these challenges, we present theory to understand the subjectivity of monoculture in Sections~\ref{sec:null_model} and~\ref{sec:population}. Finally, we propose guidelines and best practices for future evaluations on monoculture in Section~\ref{sec:discussion}.\footnote{Code can be found in \url{https://github.com/nathanaj99/subjectivity_monoculture/}. Lean verifications for proofs are given in this repository: \url{https://github.com/nikhgarg/SubjectivityOfMonocultureLean/} \cite{garg2026econcslib}.}


\section{Two Empirical Challenges}\label{sec:experiments}

We first start by demonstrating two empirical challenges, as summarized in Figure~\ref{fig:summary_null}. \textcolor{blue}{Experiment 1} shows that a minor addition to the null model can \textbf{substantially attenuate, and in some cases reverse, prior findings}. \textcolor{red}{Experiment 2} shows that inferences of monoculture also depend heavily on the population of models or questions used in the dataset. Paradoxically, \textbf{measuring monoculture is hard precisely when there truly is monoculture.} Both suggest that empirical measurements of monoculture are perhaps more fragile than prior work suggests.

\subsection{Preliminaries}

Our experiments use two notable works on monoculture as baselines: \citet{goelgreat2025} and \citet{kimcorrelated2025}. As we discuss in the prior section, these works primarily assume a \textbf{capability-based null model}. \citet{kimcorrelated2025} measure agreement rates conditional on both models being incorrect, compared to a null model where all the incorrect choices have an equal probability of being chosen.
Meanwhile, \citet{goelgreat2025} measure how much two models agree on their predicted probabilities beyond what would be expected from their marginal accuracies alone.
See Appendix~\ref{appsec:details_exp2} for details. 

In contrast to prior works, our experiments include one additional complexity in our null model: \textbf{question difficulty}. Intuitively, if two models both succeed on easy items and fail on hard ones, their agreement may exceed the capability-only baseline even if they are truly independent. More generally, our framework can handle alternative null models---e.g., models that account for topic specialization.
To operationalize this addition of question difficulty, we use a popular model of capability from psychometrics: \textit{item response theory} (IRT) \citep{lord1980applications,embretson2000item}. 

\textbf{The IRT Model.} IRT posits that each respondent $j$ possesses a latent ability vector $\theta_j$ and that each item $i$ has parameters $(a_i,d_i)$ describing the item's discrimination and difficulty, respectively. The probability of a correct response ($Y_{ij}  = 1$) is modeled as an increasing function $\sigma$ (e.g., logit or probit). We note that several recent works apply IRT to evaluating LLMs, see Appendix~\ref{appsec:related_work}.
Multidimensional IRT \citep{reckase2006multi} extends this framework by letting $\theta_j, a_i$ lie in $\mathbb R^K$. Formally, we have:
\begin{definition}[K-dimensional IRT]\label{def:mirt_null}
Fix $m\ge1$.  
For $K\in\mathbb N$, let $\mathcal N_K$ be the set of laws on $\{0,1\}^m$ generated by a \(K\)-dimensional IRT model. Conditional on item parameters
\((a_i,d_i)\) and model abilities \((\theta_j)_{j=1}^m\), responses are
independent across models and satisfy
\[
  \Pr(Y_{ij}=1\mid a_i,d_i,\theta_j)
  =
  \sigma(a_i^\top\theta_j+d_i).
\]
\end{definition}

We set $\sigma$ to be a probit function. In the context of our experiments (and in general), this model can also capture multiple choice data by encoding correctness at the answer-choice level (i.e., correctly choosing the right answer, or correctly \textit{not} choosing the wrong answer). Beyond multiple choice, we emphasize that \textbf{our insights generally hold for more complex outputs}, see Appendix~\ref{appsec:generalizations}.


We estimate the $(a_i, d_i, \theta_j)$
using standard methods, see Appendix~\ref{appsec:infer_irt}. Then, define the following:
\vspace{-0.3em}
\[
\hat p_{ij} = \Phi\!\left( \hat a_i^{\top}\hat \theta_j + \hat d_i \right), \quad R_{ij}\ :=\ Y_{ij}-\hat p_{ij},\quad
\tilde R_{ij}\ :=\ R_{ij}-\frac{1}{n}\sum_{i=1}^n R_{ij},
\]
where $\hat{p}$ is the estimated probability of correctness, $\Phi$ is the cdf of a standard normal, $R$ represents the (itemwise) residuals, and $\tilde{R}$ is the debiased $R$ to account for average calibration error in $\hat{p}$ (see Prop.~\ref{prop:residual-debiasing}).
Finally, we compute the corresponding residual correlation 
\begin{equation}
\label{eq:residual_correlation}
    \widehat{\Sigma} := E^{-\frac{1}{2}}\widehat CE^{-\frac{1}{2}}, \quad \text{where}\;\; C\ :=\ \frac{1}{n}\big(\tilde R\big)^\top \tilde R, \;\;E:=\text{diag}(\widehat C).
\end{equation}
By construction, $\widehat{\Sigma}$ has unit diagonal and is bounded in $[-1, 1]$. \textbf{$\widehat{\Sigma}$ is our measure of excess correlation or monoculture throughout the paper}. If $\widehat{\Sigma}_{k \ell} \gg 0$,
then models $k$ and $\ell$ agree more than we would expect under the null model. We emphasize that our inference procedure is not intended for out-of-sample prediction, but to infer parameters \underline{within a fixed population.} This is standard in IRT models. \textit{Conditional on a fixed population}, there is virtually no variability ($\pm$1e-4) in estimating $(a_i, d_i, \theta_j)$, and accordingly for $\widehat{\Sigma}$. 

\begin{figure*}
    \centering
    \includegraphics[width=0.9\linewidth]{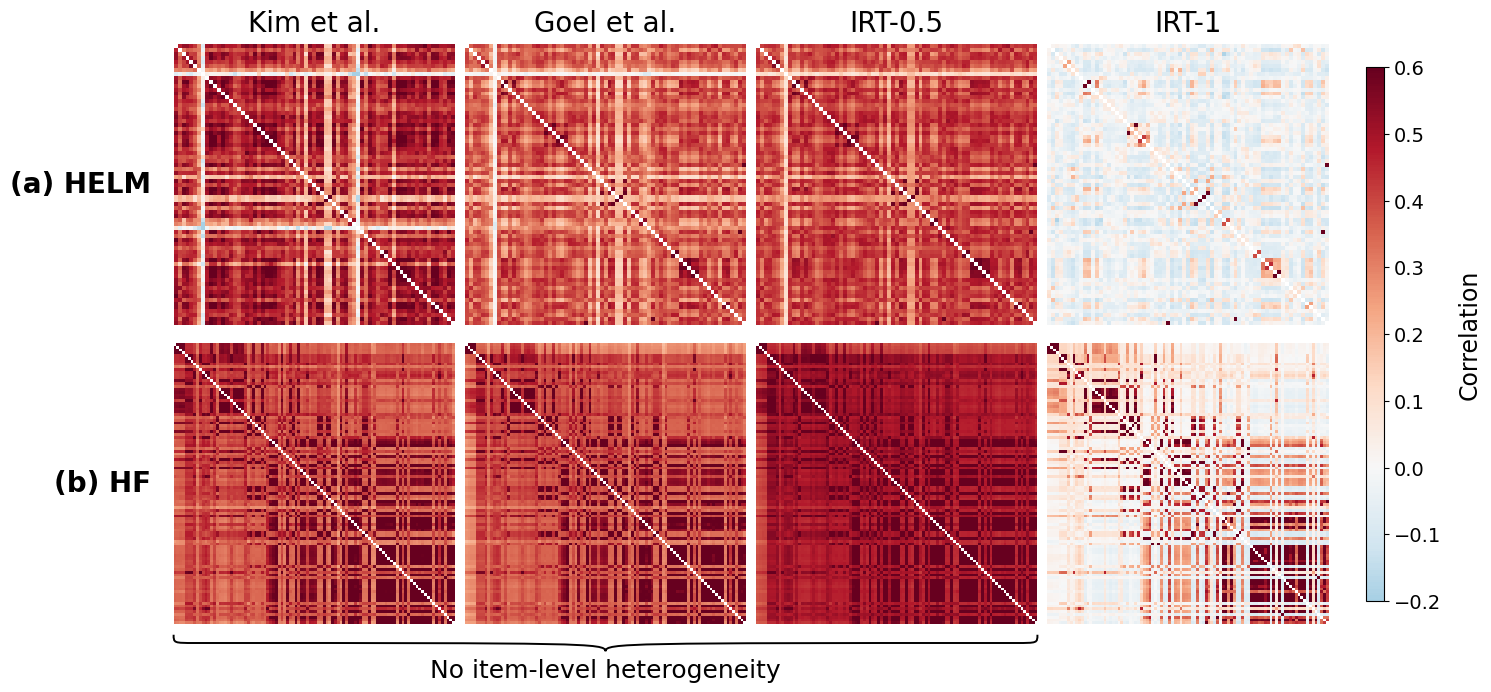}
    \caption{Residual correlation matrices for \textbf{HELM} (a) and \textbf{HF} (b) (first 100 models), using different nulls: from left to right, a baseline from \citet{kimcorrelated2025}, from \citet{goelgreat2025}, from a 1D IRT with no item difficulty, and from a 1D IRT with item difficulties. Each item is a question-answer choice pair. The first three null models \textbf{do not} include item heterogeneity, and so the full IRT model is attenuated compared to the rest because item difficulties absorb much of the apparent positive correlation.}
    \vspace{-0.5em}\label{fig:exp2_heatmaps}
\end{figure*}
\begin{figure*}
    \centering
    \includegraphics[width=\linewidth]{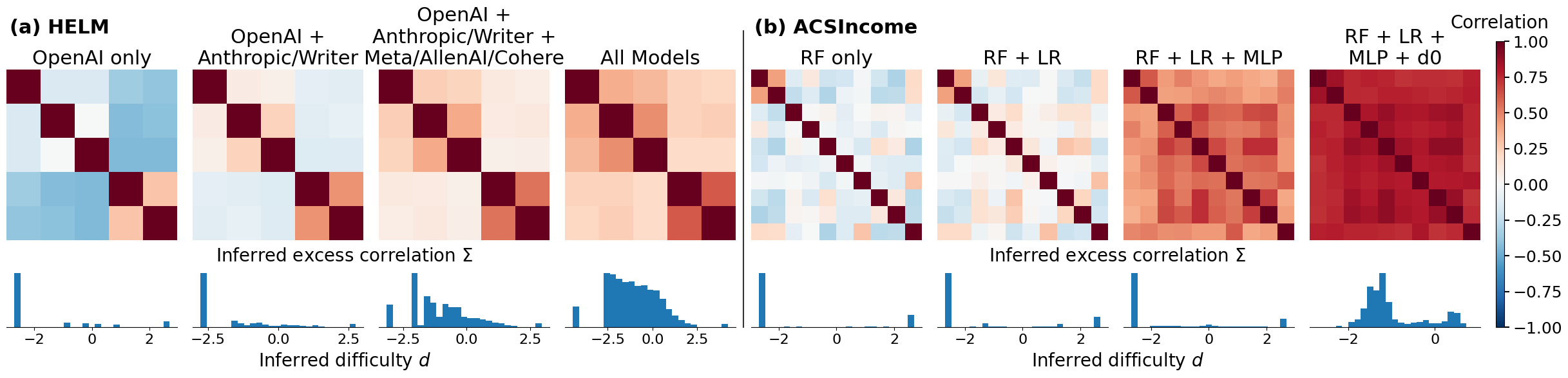}
    \caption{[Top] Inferred excess correlation matrix $\hat{\Sigma}$ from the two-stage procedure, on \textbf{HELM} (a) and \textbf{ACSIncome} (b) data. For \textbf{HELM}, the correlation matrix is shown for OpenAI models only. From left to right, more models are injected into the population for inference, starting from only OpenAI models to all models. For \textbf{ACSIncome}, the correlation matrix is shown for random forest models (RF) only, following a similar pattern as \textbf{HELM}. [Bottom] Histogram of inferred question difficulties $d$ for all questions in the dataset.}
    \vspace{-0.5em}\label{fig:heterogeneity}
\end{figure*}

\subsection{\textcolor{blue}{\Large$\bullet$} \textcolor{blue}{Experiment 1}: Incorporating item difficulty reduces excess correlation}\label{sec:comparison_previous}

\textbf{Datasets.} 
 We use two datasets following \citet{kimcorrelated2025}: \textbf{HELM}~\cite{liang2022holistic} (on MMLU data) and an Open LLM Leaderboard on Hugging Face\footnote{\url{https://huggingface.co/collections/open-llm-leaderboard/open-llm-leaderboard-2}} (hereafter, \textbf{HF}). \textbf{HELM} comprises $n=14\text{,}042$ multiple-choice questions spanning 57 academic and professional subjects, evaluated across $m=72$ language models. These include prominent proprietary models (e.g., GPT and Claude families) alongside open-weight models. \textbf{HF} contains $n=11\text{,}994$ multiple-choice questions from MMLU-Pro, a more challenging successor to MMLU, with greater emphasis on reasoning and typically 10 answer choices rather than 4. These questions are evaluated across $m=451$ predominantly open-weight LLMs.

\textbf{Note on multiple choice treatment.} The IRT model in Definition~\ref{def:mirt_null} parametrizes the probability of correctness for each question $i$. In our datasets, however, we have multiple choices that encode more specific information about \textit{when} and \textit{which types} of mistakes occur. We model this richness by expanding each unit $i$ in the IRT model as a \textbf{question-answer choice}. For example, in our experiments for \textbf{HELM}, with $n=14\text{,}042$ questions and four answer choices each, we fit an IRT model with $n=56\text{,}168$.

\textbf{Setup.} In addition to the models of \citet{kimcorrelated2025} and \citet{goelgreat2025}, we also fit two 1-dimensional IRT models. First, we consider a restricted model in which item difficulty is fixed across all items ($d_i \equiv 0$), so that variation in performance is explained solely by model ability (call this IRT-0.5). 
This regime serves as our closest analogue to the baselines, which similarly do not adjust for item-level variation. Second, we fit a full 1-dimensional IRT model with item difficulty parameters, thereby introducing item heterogeneity (IRT-1). All experiments can be executed in minutes using a CPU.

\textbf{Results.}
Figure~\ref{fig:exp2_heatmaps} reports the empirical residual correlation matrices $\widehat\Sigma$ under constant item difficulty (IRT-0.5) and item heterogeneity (IRT-1), for both \textbf{HELM} and \textbf{HF} (restricted to the first 100 models for visualization).
Relative to \citet{goelgreat2025} and \citet{kimcorrelated2025}, the correlation inferences for IRT-0.5 are slightly higher because it is relatively less expressive: IRT-0.5 only accounts for raw accuracies as a baseline, while both previous works incorporate more specific information about disagreement rates. However, crucially, the correlation inferences under IRT-1 are substantially attenuated compared to the other three, which do not model item heterogeneity---some even flipping from strongly positive to slightly negative. This result highlights the role of increasing expressivity: models may appear correlated simply because they succeed or fail on the same easy or hard questions. 

In sum, by adding a single change to the null model (question difficulty), our estimates of monoculture are significantly attenuated or reversed. The natural interpretation from Figure~\ref{fig:exp2_heatmaps} is that models dramatically less monoculture than prior works find, because their behavior on questions can be explained by questions being easy or difficult. But as we will show, there is an alternative explanation: because our null model parameters are learned from data, the \textit{set} of LLMs used impacts inferences about those parameters, and therefore, our conclusions about monoculture. We explore this in our next experiment.




\subsection{\textcolor{red}{\Large$\bullet$} \textcolor{red}{Experiment 2}: Inferences depend on the population of questions and models}\label{sec:empirical_population}

We now fix a null model (IRT-1) and vary the models used in our analysis. As we will show, greater model diversity improves inferences of monoculture, which may indicate that \textbf{HELM} suffers from poor diversity. In other words, \textbf{measuring monoculture is hard precisely when monoculture exists.} 

\textbf{Datasets.} We use two datasets: \textbf{HELM} as described above, and \textbf{ACSIncome} \cite{ding2021retiring}, which is a prediction task to determine whether an individual earns over \$50{,}000. We use this dataset because it allows finer-grained control over model classes and their inductive biases. Specifically, we train 57 models across three classes: 40 random forests (RF), 2 logistic regressions (LR), and 15 multilayer perceptrons (MLP). To generate variation within each class, we vary the number of trees, maximum depth, and splitting criterion for RF; regularization strength for LR; and hidden-layer architecture for MLP. We only train two logistic regressions because variation in parameters does not change the learned predictor. For RF and MLP, we additionally train five independently initialized instances per hyperparameter configuration, introducing variation from stochastic training. We train these models on a 50\% training split and evaluate them on the held-out 50\%. For each model, we convert held-out predictions into a binary correctness variable, which forms the dataset used in our experiment.


\textbf{Setup.}
We begin by restricting attention to a subset of models that are expected to be highly correlated due to shared inductive biases. For \textbf{ACSIncome}, we first infer $\widehat{\Sigma}$ using only the random forest models; for \textbf{HELM}, we start with only OpenAI models, which likely share underlying architectures and training.
We then progressively enlarge the population by adding other models and
re-estimating $\widehat{\Sigma}$ at each step. In \textbf{ACSIncome}, for
example, the experiment involves adding logistic regression (LR), then
multilayer perceptrons (MLP), then an artificial set of 25 models where each
model has an independent probability of success $p \sim \mathcal N(0.7, 0.2)$
(d0).
We include this final step to illustrate the case where ``ground truth''
item difficulties are all identical and independent; in this regime, we can only explain correlation via monoculture. All experiments can be executed in minutes using a CPU.

\textbf{Results.}
Figure~\ref{fig:heterogeneity} shows that when we restrict the population to a
highly correlated subset of models (left subfigures), inferences about
correlation reduce to noise. In \textbf{ACSIncome} (b), random forests (RF) do
not seem correlated despite the inductive biases shared within the model class,
since there are no additional disagreements in the dataset to contextualize the
agreements within the RF class. Intuitively, the item difficulties we learn
correspond to difficulty \textit{for RF models specifically}.
Similarly in \textbf{HELM} (a), OpenAI models do
not seem correlated despite sharing similar base architectures and training
data. These inferences change as we progressively add more diverse
models. 

In particular, we learn a more general notion of difficulty that holds across a broader class of models (see Figure~\ref{fig:heterogeneity}, bottom, which plots the distribution of $\hat{d}$). When the model population is homogeneous, many items are inferred to be very easy or very hard. Without enough disagreement across models, the null model explains agreement by reassigning it to item difficulty rather than to shared model behavior. As more diverse models are introduced, disagreements across models constrain the null fit, spreading $\hat{d}$ distribution more evenly and increasing inferred monoculture $\widehat{\Sigma}$ for the original set. 

\subsection{Discussion} 

The experiments above provide two possible explanations for the true structure of monoculture in \textbf{HELM}. \textcolor{blue}{Experiment 1} suggests that there may truly be no monoculture for \textbf{HELM} models if one believes that the IRT is an appropriate null model. Meanwhile, \textcolor{red}{Experiment 2} suggests that there may still be monoculture, but that it may be hidden by the very fact that the dataset is a monoculture.
The latter is a plausible explanation especially because \textbf{HELM} comprises only closed-source models, while \textbf{HF} includes open-source models with greater diversity in contributors and motivations behind developing the models.
At this point, we cannot disambiguate which explanation is correct.
In Sections~\ref{sec:null_model} and~\ref{sec:population}, we propose theory that provides a conceptual lens with which to understand these competing explanations. We will also conduct interpretable diagnostics in Section~\ref{sec:interpreting_monoculture} to further uncover the true structure of monoculture.

\section{\textcolor{blue}{\Large$\bullet$} Subjectivity of Null Model}\label{sec:null_model}

In \textcolor{blue} {Experiment 1}, we demonstrated that small changes in the null model can drastically change our inferences on monoculture. We now formalize what it means to assess monoculture relative to a null model. To do so, we introduce a \emph{null model of independence}: a family of joint distributions in which all dependence arises through latent parameters.

\textbf{Setup.} We observe, for each item $i \in \{1,\dots,n\}$ and model $j \in \{1,\dots,m\}$, an output $Y_{ij} \in \mathcal Y$, taken from some unknown probability distribution $P$. Throughout the paper, we instantiate our theory in the simplest setting, $\mathcal Y = \{0,1\}$, where $Y_{ij}$ encodes correctness. This choice allows us to highlight the core identification issues without additional complexity, as well as reflecting the paradigm others have studied.
\textbf{We emphasize that our insights generally hold for more complex outputs}; we show how to generalize our framework to other output spaces in Appendix~\ref{appsec:generalizations}.

\subsection{The Null Model of Independence}

We argue that claims of monoculture are only meaningful relative to a chosen \emph{null model} of independence. In the binary setting $\mathcal Y=\{0,1\}$, each model output $Y_{ij}$ is a Bernoulli random variable. Let $\Theta$ be an arbitrary parameter space encoding such latent factors (e.g., item properties, model abilities, or both). For each $\theta \in \Theta$, define a \emph{product-Bernoulli model} on $\{0,1\}^m$ by
\[
P_\theta(y_1,\ldots,y_m)
\;:=\;
\prod_{j=1}^m p_{\theta j}^{\,y_j} (1-p_{\theta j})^{1-y_j},
\]

where each $p_{\theta j} \in [0,1]$ is the conditional success probability of model $j$ under parameter $\theta$. A null model is the family of distributions $\mathcal{P}_{\text{null}} := \{P_\theta: \theta \in \Theta\}$, and a joint law $P$ is consistent with the null model if it lies in this set. In other words, under the null model, all inter-model agreement can be explained by the shared latent parameter $\theta$ together with \textbf{conditional independence across models}.
 
There are numerous ``reasonable'' choices of $\Theta$. As in \textcolor{blue} {Experiment 1}, we demonstrated that a 
simpler null model that only accounts for accuracy might imply a higher observed correlation than if we additionally included item difficulties. In other words, expanding $\Theta$ enlarges $\mathcal{P}_{\text{null}}$.
In the extreme, if $\Theta$ is allowed to be arbitrarily rich, then $\mathcal{P}_{\text{null}}$ can express any distribution on $\{0,1\}^m$. The
next theorem formalizes this observation.\footnote{Note that this finite mixture representation is a simple instance of a more general phenomenon: exchangeable or symmetric laws can often be written as mixtures of independent structures, as in de Finetti-type theorems.} All proofs can be found in Appendix~\ref{appsec:thm_proofs}.

\begin{theorem}\label{thm:mixture}
For any probability distribution $P$ on $\{0,1\}^m$, there exists a probability measure $H$ on $[0,1]^m$ and a latent vector $P_i=(P_{i1},\ldots,P_{im})\sim H$ such that, conditional on $P_i$, the coordinates are independent Bernoulli:
\[
Y_{ij}\mid P_i \ \indep \ (j=1,\ldots,m),\qquad \mathbb P(Y_{ij}=1\mid P_i)=P_{ij},
\]
and the unconditional distribution of $Y_i$ equals $P$:
\[
\mathbb P(Y_i=y)\;=\;\int \prod_{j=1}^m p_j^{y_j}(1-p_j)^{1-y_j}\,dH(p_1,\ldots,p_m).
\]

Moreover, there exists a discrete $H$ supported on at most $2^m$ points with this property.
\end{theorem}

In other words, there exists a sufficiently expressive null that makes the data look independent conditional on the model parameters $\Theta$. From hereon in this section, we suppress $\Theta$ and work directly with abstract families of distributions. All results apply uniformly over any choice of $\Theta$.

\subsection{Absorption of Excess Correlation}
\label{sec:absorption}



Theorem~\ref{thm:mixture} hints at the central tension of measuring monoculture: \textbf{the more structure we allow the null model to capture, the less apparent correlation remains to be explained.} 
If the null is too simple, it under-explains the data and makes monoculture appear everywhere; if the null is too rich, it explains away all dependence and makes monoculture impossible to detect. To formalize this idea, we introduce a ``null ladder’’ in Definition~\ref{def:null_ladder}: a nested sequence of
increasingly expressive null models. As we climb the ladder, the null becomes more permissive and the gap
between the observed law $P$ and the best-fitting null model must shrink.

\begin{definition}[Null ladder]\label{def:null_ladder}
Let $\mathcal P$ be the set of all laws on $\{0,1\}^m$.
A \emph{null ladder} is a nested sequence $(\mathcal N_K)_{K\ge1}$ with
\[
\mathcal N_1 \subseteq \mathcal N_2 \subseteq \cdots \subseteq \mathcal N_K \subseteq \mathcal N_{K+1} \subseteq \cdots \subseteq \mathcal P.
\]
\end{definition}
\begin{definition}[Excess at level $K$]\label{def:excess_k}
Let $D:\mathcal P\times\mathcal P\to[0,\infty)$ be a discrepancy with $D(P,Q)\ge0, D(P,P)=0$, such as TV distance or $f$-divergence.
For $P\in\mathcal P$, define the \emph{minimal excess} at level $K$ and the \emph{best-fitting null}:
\[
Q_K^\star \in \arg\inf_{Q\in \mathcal N_K} D(P,Q), \quad E^P_{\min}(K)\ :=\  D(P,Q^\star_K).
\]

When the infimum is not attained, let
$(Q_K^\star)_{K\ge1}$ be any sequence in $\mathcal N_K$ s.t.
$D(P,Q_K^\star)\to E^P_{\min}(K)$.
\end{definition}

The quantity $E^P_{\min}(K)$ measures how much of $P$ remains unexplained at the $K$th
level of the ladder. With this quantity, we can now state our next result, which simply formalizes the intuition that a more expressive model can only explain \emph{more} of the data, never less.

\begin{proposition}
For any null ladder $(\mathcal N_K)$ and $P\in\mathcal P$,
$\cdots \ge\ E^P_{\min}(K) \ge E^P_{\min}(K+1)\ge \cdots
$.
\label{prop:monotone_absorption}
\end{proposition}

Proposition~\ref{prop:monotone_absorption} focuses on the global discrepancy between $P$ and the null, rather than \textit{cross-model correlations} unexplained by the null, which we formalize next. Theorem~\ref{thm:residual-dependence-vanishes} shows that, provided $P$ can be approximated arbitrarily well along the ladder, these residual covariances vanish as $K$ increases. In other words, a sufficiently rich null model can always reinterpret what appears to be cross-model correlation as arising from latent structure rather than monoculture. 

\begin{theorem}
\label{thm:residual-dependence-vanishes}
Fix a parametrized null ladder $(\mathcal N_K)_{K \geq 1}$ and let $Q_K^\star$ be a best-fitting null law at level $K$ as in Definition~\ref{def:excess_k}.
Suppose that $D(P,Q_K^\star)\xrightarrow[K\to\infty]{}0,$ where $D$ is any discrepancy that controls expectations of bounded functions,
such as total variation.

For each $K$, let $p^{(K)}_{ij} \ :=\ \Pr_{Q_K^\star}(Y_{ij}=1)$ denote the null-implied conditional
success probability for item \(i\) and model \(j\) under \(Q_K^\star\). Then, for every fixed pair of distinct models $j \neq \ell$,

\[
\left|\frac{1}{n} \sum_{i = 1}^n\E_{P}\!\Big[(Y_{ij}-p^{(K)}_{ij})(Y_{i\ell}-p^{(K)}_{i\ell})\Big]\right|
\ \xrightarrow[K\to\infty]{}\ 0.
\]
\end{theorem}




\textbf{Connection to Empirical Results.} One type of null ladder is IRT-0.5 and IRT-1 from \textcolor{blue}{Experiment 1}. IRT-1 (varied item difficulty) subsumes IRT-0.5 (constant item difficulty) in complexity and is therefore able to absorb more of the pairwise correlations as per Thm~\ref{thm:residual-dependence-vanishes}. Multi-dimensional IRT is another null ladder, and in Figure~\ref{fig:off_diag_main} we show that $\widehat{\Sigma}$ tend to zero as dimension $K$ increases. See Appendix~\ref{appsec:irt_k_exp} for details.

\begin{figure*}
    \centering
    \includegraphics[width=0.85\linewidth]{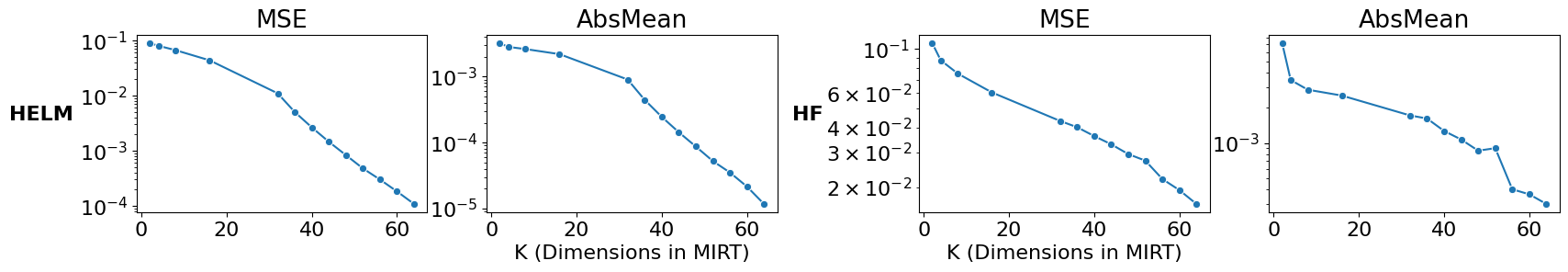}
    \caption{Mean square error and absolute mean of the residual correlations $\widehat{\Sigma}$, as a function of $K$: dimensions in the multidimensional IRT model, i.e., increading complexity of null model. Left (right) shows results on the \textbf{HELM} (\textbf{HF}) dataset. As $K$ increases, residual correlations that are unexplained by the null model tend toward zero, meaning that increasingly expressive null models can absorb model correlations.}

\label{fig:off_diag_main}
\end{figure*}

\section{\textcolor{red}{\Large$\bullet$} Relativity of Population}\label{sec:population}

In \textcolor{red}{Experiment 2} (Section~\ref{sec:empirical_population}), we demonstrated that inferences of monoculture are less reliable precisely when there truly is monoculture. We now theoretically formalize these observations.  Intuitively, greater diversity results in more informative disagreements between models that better reveal the underlying correlation structure. 

\subsection{Preliminaries}
For notational clarity, we introduce a more specific parametrization of the null model $\Theta = (U,V)$:
\begin{equation}\label{eq:generic-null}
Y_{ij}\mid (U_i,V_j)\ \indep\ \forall j,
\quad 
\Pr(Y_{ij}=1\mid U_i,V_j)=f(U_i,V_j),
\end{equation}
where $U_i\in\mathbb R^r$ are item factors, $V_j\in\mathbb R^r$ are model factors,
$f:\mathbb R^r\times\mathbb R^r\to(0,1)$ is a fixed link satisfying mild
regularity, and $r\in\mathbb N$ is fixed. Note that this includes the IRT models as special cases.

For any choice of item set $I$ and model set $J$, let $P_{I,J}$ denote the joint distribution of model responses on items $I$ and models $J$. Let $\mathcal N_r(I,J)$ denote the set of distributions that can be explained by an $r$-dimensional null model \eqref{eq:generic-null} using only items $I$ and models $J$.

\begin{definition}[Population-specific target]
    Given any discrepancy measure $D(P,Q)$ (e.g.\ TV or an $f$-divergence), define the \emph{best-fitting null} as \begin{equation}\label{eq:proj-null}
    Q^\star_{I,J}\ \in\ \arg\min_{Q\in\mathcal N_r(I,J)} D(P_{I,J},Q).
    \vspace{-0.7em}
    \end{equation}
\label{def:population_target}
\end{definition}

Under Definition~\ref{def:population_target}, any notion of ``excess correlation'' built from residuals under $Q^\star_{I,J}$ is a function of $(I,J)$. As before, we use Eq.~\eqref{eq:residual_correlation} as our notion of excess correlation.

\subsection{Inferences depend on the population}

It is almost immediate that changing the population changes the null fit, which also changes the residual covariances, as formalized in the following proposition. 
In this sense, notions such as difficulty, ability, and excess correlation are not intrinsic properties of isolated models or items, but are defined relative to the population under consideration.

\begin{proposition}[Population relativity of the null fit]\label{prop:relativity}
Fix $r$ and $D$. The optimizer $Q^\star_{I,J}$ in \eqref{eq:proj-null} may vary as $I, J$ change: there may exist item/model sets $(I,J)$ and $(I',J')$ with $I'\subseteq I$ and $J'\subseteq J$ such that the induced marginals do not match: $\bigl(Q^\star_{I,J}\bigr)\big|_{(I',J')}\ \neq\ Q^\star_{I',J'}.$

Here, $\bigl(Q^\star_{I,J}\bigr)\big|_{(I',J')}$ denotes the marginal predictions for items $I'$ and models $J'$ induced by a null model fit on the larger population $(I, J)$.
\end{proposition}



What is the ``right'' population to choose? For example, if all questions are very similar in difficulty, then we cannot expect to learn much about model ability. Similarly, if all models behave nearly identically, we cannot distinguish whether they are truly homogeneous or simply being evaluated on trivially easy or hard questions. This intuition suggests that diversity is useful, as we formalize next.




\textbf{Heterogeneity measure.}
Let $(U^\star,V^\star)$ be any fitted null representation achieving $Q^\star_{I,J}$ in Defn.~\ref{def:population_target}.
Define summary covariance matrices $G$ of the fitted item and model factors and heterogeneity $h$:

\[
G_U^\star := \sum_{i\in I} U_i^\star (U_i^\star)^\top,
\quad
G_V^\star := \sum_{j\in J} V_j^\star (V_j^\star)^\top, \quad
h(I,J)\ :=\ \min\bigl\{\lambda_{\min}(G_U^\star),\ \lambda_{\min}(G_V^\star)\bigr\}.
\]
Here, $\lambda_{\min}(\cdot)$ denotes the smallest eigenvalue, which measures the least amount of variation along any direction in the latent space. Thus, $h(I, J)$ captures the weakest direction of variation across items or models under the fitted null. Larger $h(I,J)$ means that items and models explore more independent directions in the latent space (i.e., more heterogeneous). 

We note that $h$ is defined using latent factors, so equivalent reparametrizations can change $h$.\footnote{E.g., scaling $U^\star$ by $s$ and $V^\star$ by $1/s$ leaves every $p^\star_{ij}$ unchanged but sends $h\to0$ as $s\to\infty$.} A principled way to remove this dependence is to evaluate $h$ at the representation that maximizes it, and we take $(U^\star,V^\star)$ to be this representation henceforth. For example, for unrestricted MIRT-type links $f(u,v)=g(u^\top v)$, a canonical maximizer is the balanced representation ($G_U^\star=G_V^\star$), which is unique up to rotation and gives $h(I,J)=\sigma_r(\Theta^\star)$, the $r$-th singular value of the fitted logit matrix $\Theta^\star_{ij}=g^{-1}(p^\star_{ij})$. Thus, in the unrestricted model, $h$ depends only on the fitted probabilities, and hence only on the population $(I,J)$. As we show in the next result, heterogeneity improves the local stability (i.e., conditioning) of the null fit.


\begin{theorem}
\label{thm:heterogeneity}
Let $p^\star_{ij}=f(U_i^\star,V_j^\star)$ be the fitted
probabilities, where $f$ is twice continuously differentiable.
Assume that the null model is locally identifiable up to
reparametrization, with identification strength bounded below
proportionally to the heterogeneity index $h>0$
(\Cref{ass:local-id}).

For $\varepsilon>0$, define the local
\emph{$\varepsilon$-equivalence set} of latent representations by
\[
\mathcal E_\varepsilon
:=
\Bigl\{(U,V)\in\mathcal N:
\sum_{i\in I}\sum_{j\in J}
\bigl(f(U_i,V_j)-p^\star_{ij}\bigr)^2
\le \varepsilon^2
\Bigr\},
\]
where $\mathcal N$ is a sufficiently small neighborhood
of the fitted representation.
Then $\exists$ a constant $C>0$ such that, for all
sufficiently small $\varepsilon$,
\[
\sup_{(U,V)\in\mathcal E_\varepsilon}
\operatorname{dist}\bigl((U,V),\mathcal E_0\bigr)
\le \frac{C}{\sqrt{h}}\,\varepsilon,
\]
where $\operatorname{dist}$ is Euclidean distance in $(U,V)$ and
$\mathcal E_0:=\{(U,V): f(U_i,V_j)=p^\star_{ij}\ \forall i,j\}$ is the set of
representations that reproduce the fitted probabilities exactly.
\end{theorem}

We measure distance to $\mathcal E_0$ rather than to $(U^\star,V^\star)$ because
null models have exact invariances: e.g., for an IRT link $f(u,v)=g(u^\top v)$,
replacing $(U_i,V_j)$ by $(AU_i,A^{-\top}V_j)$ for any invertible $A$ leaves every
probability unchanged. In that case, $\mathcal E_0$ consists exactly of these
rescalings and rotations of $(U^\star,V^\star)$, and so our theorem bounds the fit \emph{after optimally re-aligning its latent coordinates}. Note that \Cref{ass:local-id} holds automatically for IRT/MIRT, see Lemma~\ref{lem:irt-id}.

Intuitively, if many items (or models) are similar, the null model has substantial freedom: it can shift latent factors across similar directions without meaningfully changing its predicted probabilities. Since our inferences about monoculture depend on these latent factors, the fact that there exists many factor representations that yield similar observed distributions means that we cannot make strong claims about monoculture.
However, when items and models are diverse, this flexibility disappears: the null parameters become locally well-identified, and the fitted representation is more stable. Thus, having datasets and models with heterogeneity is itself helpful in measuring monoculture.

\textbf{Connection to Empirical Results.} In the regime where models are homogeneous from \textcolor{red}{Experiment 2}, Thm~\ref{thm:heterogeneity} suggests that many possible $(U, V)$ values can plausibly explain the data. Our particular inference procedure happened to explain the observed agreement by pushing question difficulties to be extremely hard or easy, rather than recovering the true correlation structure. However, greater diversity improves
the conditioning of our inferences as per Thm~\ref{thm:heterogeneity}.

\begin{figure*}
    \centering
    \includegraphics[width=0.95\linewidth]{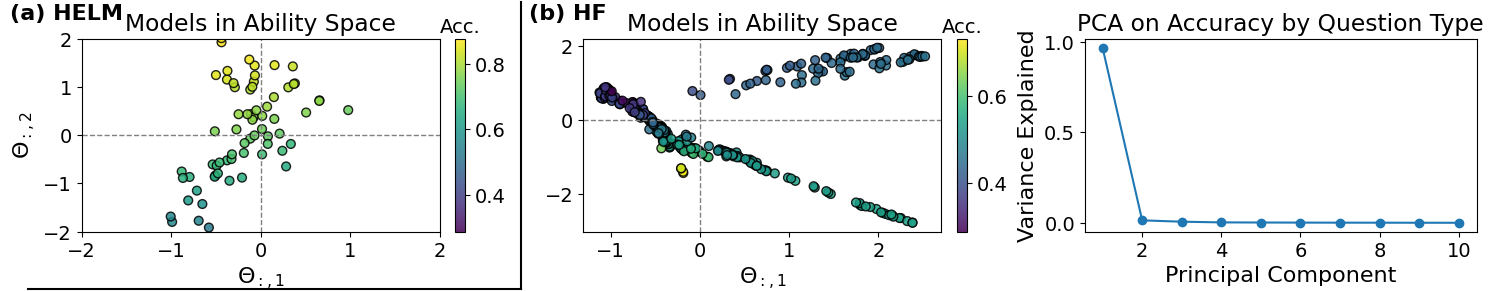}
    \includegraphics[width=0.95\linewidth]{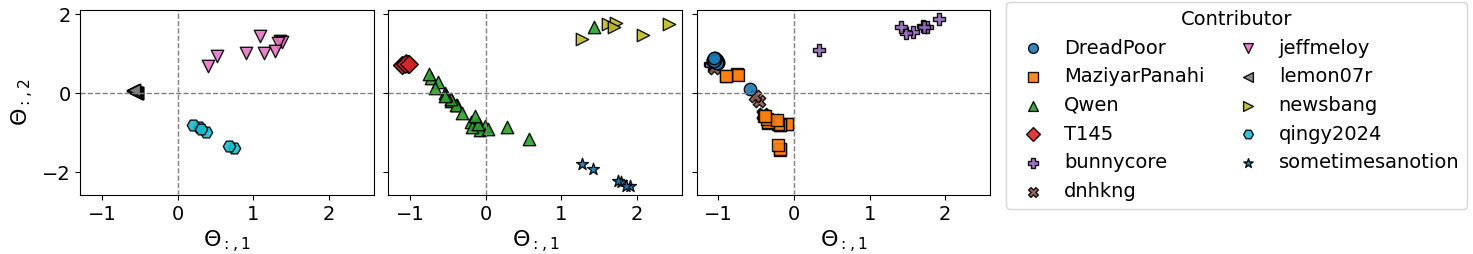}
    \caption{(a) Scatter plot of inferred model ability $\Theta$ in a 2-D IRT model, colored by accuracy over entire benchmark dataset (\textbf{HELM}). (b; top middle) Same as (a), but for \textbf{HF} data. (b; top right) Percentage of variance explained by principal component, when applying PCA on model accuracy stratified by question type. (b; bottom) Scatter plot of inferred ability $\Theta$, disaggregated by contributor.}
    \vspace{-0.5em}
\label{fig:irt_k2}
\end{figure*}

\section{What should we do in practice?}\label{sec:discussion}
\vspace{-0.5em}

We now provide guidelines for evaluators to better analyze monoculture in practice.

\begin{tcolorbox}[title=Guidelines,
breakable
]
\begin{enumerate}[leftmargin=0.2cm]
    \item \textcolor{olive}{\textsf{Select and justify the following choices:}}

    (I) \textbf{The set of models and questions to ensure a diverse and representative sample.} For models, use additional information about contributors and architectures. For questions, ensure that they encompass the various skills within the scope of the benchmark.

    (II) \textbf{The set of reasonable null models.} Deciding a null baseline should be a \textit{theory-driven process} rather than a data-driven one, informed by one's beliefs about the data generating process and how models behave. Evaluators should decide \textbf{what variation should the null be allowed to explain away.} Should agreement induced by overall model ability, item difficulty, or shared architectures and training paradigms be treated as expected or as part of the phenomenon of interest? This decision involves various factors, such as: 

    \begin{enumerate}
    \item \textit{What level of granularity is appropriate?}
    \item \textit{What assumptions of exchangeability or conditional independence are defensible?}
    \item \textit{How much flexibility should the null have?}
    \end{enumerate}

    \item \textcolor{olive}{\textsf{Conduct robustness checks:}} \textbf{Compare monoculture results across the set of reasonable null models.} If the results do not change across null models, then the conclusions are defensible and valid. If results differ (as in Section~\ref{sec:comparison_previous}), go to Step 3. This principle is similar to best practices in empirical work, where robustness checks test whether or not results are valid against different assumptions about the world/data generating process.

    \item \textcolor{olive}{\textsf{Explain discrepancies:}} \textbf{Run appropriate diagnostics for further insights on correlation structure.} See the next section for an example of finding interpretable explanations.
\end{enumerate}

\centering
\includegraphics[width=\linewidth]{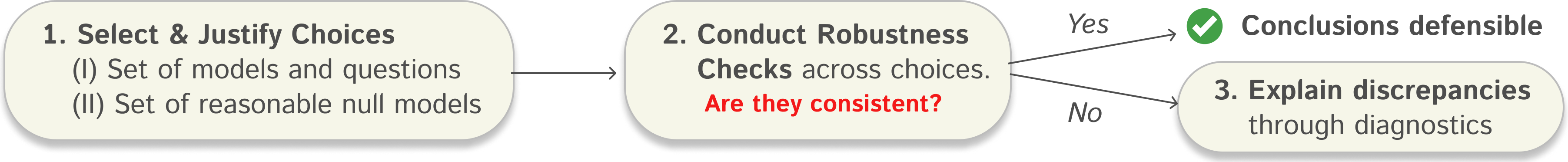}
\end{tcolorbox}

\vspace{-0.8em}

\subsection{Our empirical findings as an exercise}\label{sec:interpreting_monoculture}

\textcolor{blue}{Experiment 1} and \textcolor{red}{Experiment 2} are instances of the first two steps of this guideline. We first selected two reasonable null models and several sets of models (\textcolor{olive}{\textsf{Step 1}}). Then, we inferred monoculture based on these decisions as robustness checks (\textcolor{olive}{\textsf{Step 2}}). However, we found substantial inconsistencies, especially between our two null models. What remains to be done is \textcolor{olive}{\textsf{Step 3}} (explaining discrepancies). One idea to explain discrepancies is to find interpretable structures among models. To do so, we fit a $K=2$ IRT model to both \textbf{HELM} and \textbf{HF} and inspect the geometry of the resulting latent spaces. 

Our results are shown in Figure~\ref{fig:irt_k2}. For \textbf{HELM}, we find that model abilities collapse onto a single dimension that is highly correlated with overall accuracy (subfig (a)), indicating that the 1D IRT may be a reasonable descriptive model of the dataset. However, this finding is not evidence that \textbf{HELM} models truly behave independently, as one would conclude from Figure~\ref{fig:exp2_heatmaps}(a, right). Rather, it suggests that \textbf{HELM} models have limited coverage for detecting dependence beyond a single ability dimension. If \textbf{HELM} models are all produced by similar training pipelines or architectures, then we would need additional diversity (e.g., open-source models) in order to reveal what kinds of questions only this family of models systematically gets right or wrong. In other words, while the 1D IRT is sufficient in characterizing the models in \textbf{HELM}, we still cannot answer if there is truly monoculture.

\textbf{HF} (subfig (b)) presents a more definitive picture. Model abilities are more dispersed in the two-dimensional latent space and form two distinct clusters that are not strongly correlated with overall accuracy (top middle plot). A natural hypothesis is that these clusters correspond to specialization (e.g., mathematical ability vs reading comprehension). However, a PCA analysis of model accuracies stratified by question type does not support this explanation (top right plot). Instead, the clustering appears to align more closely with
model provenance, including common contributors, model families, or development
pipelines (bottom plots). This suggests that a 2D (or higher) IRT model is \textbf{not appropriate} because the additional structure may be reflected by other sources of heterogeneity like model provenance. We can therefore conclude that there exists some monoculture structure  (as in Fig~\ref{fig:exp2_heatmaps}, right), because it cannot be otherwise explained by an appropriate richer null model.

\subsection{Conclusion}
We have shown that inferences of monoculture are sensitive to two key choices: \textit{(1)} the baseline null model (Fig~\ref{fig:exp2_heatmaps}, Section~\ref{sec:null_model}); and \textit{(2)} the population of models and questions (Fig~\ref{fig:heterogeneity}, Section~\ref{sec:population}). This subjectivity means that robustness checks should be a best practice to ensure that claims of monoculture are consistent across various modeling assumptions, as we establish in Section~\ref{sec:discussion}.

\textbf{Limitations \& Future Work.} While our theoretical framework extends to richer outputs like free-form text, it remains unclear what a suitable null model would look like in these settings. Building such models and empirically testing for monoculture are promising directions for future work.

\section*{Acknowledgments}
NG is supported by NSF CAREER IIS-2339427, NASA, the William T. Grant Foundation, and research awards from the Cornell Tech Urban Tech Hub, Google, Mastercard, and Amazon. NJ and MR are supported by the MIT Generative AI Impact Consortium (MGAIC). All opinions are those of the authors and not of any funding agencies, supporters, or other institutions.




\bibliography{bib}
\bibliographystyle{plainnat}

\newpage
\appendix

\section{Other Related Work}\label{appsec:related_work}
\textbf{Model Evaluations using IRT.} While our results apply broadly to any choice of null models, in this work we often use IRT models as a null because they incorporate item-level heterogeneity. Many recent works have used IRT in their evaluations \cite{martinez2019item, rodriguez2021evaluation, polo2024tiny, schilling2025lifting, zhou2025lost, castleman2025rethinking}, acknowledging that questions vary in characteristics and some may be more informative than others. However, all of these works use IRT in the context of evaluating model capabilities rather than, as we do, inter-model correlations.

\textbf{Model Multiplicity.} Model multiplicity suggests that highly expressive models admit many different solutions that are consistent with the same training data, often reflecting different inductive biases or arbitrary choices among equivalent hypotheses \cite{black2022model, marx2020predictive}. There have been many works that seek to characterize or mitigate multiplicity across a variety of contexts \cite{pawelczyk2020counterfactual, rodolfa2021empirical, coston2021characterizing,d2022underspecification,semenova2022existence, black2022consistent, black2022selective}. In these works, multiplicity is typically framed as a property of the \textit{model class}: more flexible hypotheses allow for a greater number of solutions. However, our results suggest another perspective in which multiplicity is also shaped by \textit{the population of models used for evaluation}. We discuss this further in Section~\ref{sec:discussion}.

\section{Additional Theorems and Proofs}\label{appsec:thm_proofs}
\subsection{Additional Results}

\begin{lemma}[Nestedness of ability nulls]\label{lem:nested_mirt}
For all $K$, we have $\mathcal N_K \subseteq \mathcal N_{K+1}$.
\end{lemma}

\begin{proof}
Embed the $K$-dimensional latent $\theta_i$ into $\mathbb R^{K+1}$ by
$\theta_i'=(\theta_i,0)$ and extend $a_j'=(a_j,0)$.  
Then $\sigma(a_i'{}^\top\theta_j'+b_i)=\sigma(a_i^\top\theta_j+b_i)$,  
so every law representable at dimension $K$ is also representable at $K{+}1$.
\end{proof}

\begin{proposition}[Debiasing residuals removes calibration error]
\label{prop:residual-debiasing}
Let $Y_{ij}\in\{0,1\}$ denote the observed response for item $i\in\{1,\dots,n\}$ and model/respondent $j\in\{1,\dots,m\}$. 
Assume the item rows $Y_i=(Y_{i1},\dots,Y_{im})$ are i.i.d.\ from some law $P$.
For each $(i,j)$, let $p^{\star}_{ij} := \Pr(Y_{ij}=1)$ denote the \emph{true} success probability and let $\hat p_{ij}\in(0,1)$ be any (possibly misspecified) model-based estimate of $p^{\star}_{ij}$.

Define the raw residual
\[
e_{ij} \;:=\; Y_{ij} - \hat p_{ij},
\]
and, for a fixed model $j$, define the sample mean residual across items
\[
\bar e_{\cdot j} \;:=\; \frac{1}{n} \sum_{i=1}^n e_{ij}
\;=\;
\frac{1}{n} \sum_{i=1}^n (Y_{ij} - \hat p_{ij}).
\]

Then the population mean of the raw residual equals the population calibration error:
\[
\mathbb{E}[e_{ij}] 
\;=\; 
\mathbb{E}[Y_{ij} - \hat p_{ij}]
\;=\;
\mathbb{E}[p^{\star}_{ij} - \hat p_{ij}].
\]

Since the response rows are i.i.d., writing $p_j=\Pr(Y_{ij}=1)$,
\[
\bar e_{\cdot j}-\frac1n\sum_i(p_j-\hat p_{ij})
=\frac1n\sum_iY_{ij}-p_j\xrightarrow{p}0
\qquad\text{as } n\to\infty.
\]
Thus the mean residual tracks the empirical average calibration error, even when the fitted probabilities depend on the sample.

Consequently, the \emph{debiased} residual
\[
\tilde e_{ij} \;:=\; e_{ij} - \bar e_{\cdot j}
\;=\; (Y_{ij} - \hat p_{ij}) - \frac{1}{n}\sum_{i'=1}^n (Y_{i'j} - \hat p_{i'j})
\]
satisfies
\[
\frac1n\sum_{i=1}^n\tilde e_{ij}=0
\]
exactly. Moreover,
\[
\tilde e_{ij}
=\xi_{ij}-\bar\xi_{\cdot j}+\delta_{ij}-\bar\delta_{\cdot j},
\]
i.e.\ it removes the average calibration error but
retains centered itemwise misspecification.

In particular, sample covariances (or correlations) computed from the debiased residuals $\tilde e_{ij}$ estimate the \emph{residual dependence} across models, relative to the fitted probabilities; demeaning removes average calibration error but may leave dependence induced by item-specific misspecification
\end{proposition}

\subsection{Proofs}
\paragraph{Proof of Theorem~\ref{thm:mixture}.}
\begin{proof}
Define a discrete probability measure $H^\star$ on the vertices
$\{0,1\}^m \subset [0,1]^m$ by placing mass $P(y)$ at the point $p=y$:
$H^\star(\{y\}) = P(y)$ for each $y \in \{0,1\}^m$.
Given $P_i \sim H^\star$, define the conditional law of $Y_i$ by
independent Bernoulli coordinates with success probabilities $P_{ij}$:
\[
\mathbb P(Y_{ij} = 1 \mid P_i = p) = p_j,
\qquad j=1,\dots,m.
\]
If $P_i = y \in \{0,1\}^m$, then each $Y_{ij}$ is almost surely equal to $y_j$,
so
\[
\mathbb P(Y_i = y \mid P_i = y) = \prod_{j=1}^m
y_j^{y_j}(1-y_j)^{1-y_j} = 1.
\]
Integrating over $H^\star$ gives, for any $y \in \{0,1\}^m$,
\[
\mathbb P(Y_i = y)
= \int \mathbb P(Y_i = y \mid P_i = p)\, dH^\star(p)
= \sum_{y' \in \{0,1\}^m} \mathbb I\{y'=y\} P(y')
= P(y).
\]
Thus $P$ admits the desired mixture representation with $H = H^\star$.
By construction, $H^\star$ is supported on at most $2^m$ points.
\end{proof}


\paragraph{Proof of Proposition~\ref{prop:monotone_absorption}.}

\begin{proof}
Since the null ladder is nested, we have $\mathcal N_K \subseteq \mathcal N_{K+1}$ for every $K$.
Therefore, minimizing over $\mathcal N_{K+1}$ gives at least as much flexibility as minimizing over
$\mathcal N_K$. Hence,
\[
E^P_{\min}(K+1)
= \inf_{Q\in \mathcal N_{K+1}} D(P,Q)
\le
\inf_{Q\in \mathcal N_K} D(P,Q)
= E^P_{\min}(K).
\]
Thus $E^P_{\min}(K)$ is nonincreasing in $K$.
\end{proof}

\paragraph{Proof of Theorem~\ref{thm:residual-dependence-vanishes}.}
\begin{proof}
Fix $K$ and distinct $j\neq \ell$. Define the function $g_K:\{0,1\}^{n\times m}\to\mathbb R$ by
\[
g_K(y)\ :=\ \frac{1}{n}\sum_{i=1}^n (y_{ij}-p^{(K)}_{ij})(y_{i\ell}-p^{(K)}_{i\ell}).
\]
Since each $y_{ij}\in\{0,1\}$ and $p^{(K)}_{ij}\in[0,1]$, we have $|y_{ij}-p^{(K)}_{ij}|\le 1$, hence
\[
|g_K(y)| \le \frac{1}{n}\sum_{i=1}^n 1 = 1
\quad\text{for all } y,
\]
so $\|g_K\|_\infty \le 1$.

By definition of the null model in Section~\ref{sec:null_model}, for each item $i$,
\[
\E_{Q_K^\star}\!\Big[(Y_{ij}-p^{(K)}_{ij})(Y_{i\ell}-p^{(K)}_{i\ell})\Big]
=
\E_{Q_K^\star}[Y_{ij}-p^{(K)}_{ij}]\;\E_{Q_K^\star}[Y_{i\ell}-p^{(K)}_{i\ell}]
=0,
\]
because $\E_{Q_K^\star}[Y_{ij}-p^{(K)}_{ij}]=\Pr_{Q_K^\star}(Y_{ij}=1)-p^{(K)}_{ij}=0$.
Averaging over $i$ gives $\E_{Q_K^\star}[g_K(Y)] = 0$.

Therefore,
\[
\left|\E_{P}[g_K(Y)]\right|
=
\left|\E_{P}[g_K(Y)]-\E_{Q_K^\star}[g_K(Y)]\right|.
\]
Now use the standard total variation inequality: for any two probability laws $\mu,\nu$
on a finite measurable space and any $f$ with $\|f\|_\infty\le 1$,
\[
\big|\E_\mu[f]-\E_\nu[f]\big| \le 2\,\mathrm{TV}(\mu,\nu).
\]
Applying this with $\mu=P$, $\nu=Q_K^\star$, and $f=g_K$ yields
\[
\left|\E_{P}[g_K(Y)]\right| \le 2\,\mathrm{TV}(P,Q_K^\star).
\]
Finally, note that $\E_P[g_K(Y)]$ is exactly the left-hand side quantity:
\[
\E_P[g_K(Y)]
=
\frac{1}{n} \sum_{i = 1}^n\E_{P}\!\Big[(Y_{ij}-p^{(K)}_{ij})(Y_{i\ell}-p^{(K)}_{i\ell})\Big].
\]
The claimed bound follows, and since $\mathrm{TV}(P,Q_K^\star)\to 0$ by assumption, the left-hand side converges to $0$.
\end{proof}




\begin{remark}
\label{rem:empirical-residuals}
Theorem~\ref{thm:residual-dependence-vanishes} is stated in terms of the model-implied probabilities $p_{ij}^{(K)}$. In practice we use fitted probabilities $\hat p_{ij}^{(K)}$ obtained from the $K$-dimensional IRT fit. If the fitted probabilities are uniformly consistent for the
model-implied probabilities and the corresponding empirical residual moments obey a law of large numbers in the stated sample-size regime, then the corresponding sample residual
covariances inherit the same limiting zero-residual conclusion. The conclusion also holds for correlations if the corresponding residual variances stay bounded away from zero.
\end{remark}

\paragraph{Proof of Proposition~\ref{prop:relativity}.}
\begin{proof}
The key observation is that the best-fitting null model is defined \emph{relative to the population} through the optimization problem \eqref{eq:proj-null}. Changing the item or model set changes both the feasible set $\mathcal N_r(I,J)$ and the objective being minimized.

Fix $(I,J)$ and suppose $Q^\star_{I,J}$ is a minimizer of $D(P_{I,J},Q)$ over $\mathcal N_r(I,J)$. By construction, the parameters $(U,V)$ underlying $Q^\star_{I,J}$ must jointly explain \emph{all} item--model interactions in $(I,J)$ using a shared $r$-dimensional representation. In particular, the fitted parameters for items in $I'$ and models in $J'$ are influenced by the presence of additional items $I\setminus I'$ and models $J\setminus J$ through this shared low-dimensional constraint.

Now consider fitting the null model directly on the restricted population $(I',J')$. The optimizer $Q^\star_{I',J'}$ is free to choose parameters that best explain $P_{I',J'}$ alone, without needing to simultaneously accommodate the behavior of items or models outside this subset. Unless the restriction of $Q^\star_{I,J}$ to $(I',J')$ already happens to be optimal for this smaller problem—which would require a special compatibility condition—there is no reason for the induced marginal $\bigl(Q^\star_{I,J}\bigr)\big|_{(I',J')}$ to coincide with $Q^\star_{I',J'}$.

Therefore, there exist populations $(I,J)$ and strict subsets $(I',J')$ for which
\[
\bigl(Q^\star_{I,J}\bigr)\big|_{(I',J')} \neq Q^\star_{I',J'},
\]
establishing that the null fit, and hence any residual-based notion of excess correlation, depends on the chosen population.
\end{proof}

\paragraph{Local identifiability.}
Write $\Theta=(U,V)$, $\|\Theta\|^2:=\sum_{i\in I}\|U_i\|^2+\sum_{j\in J}\|V_j\|^2$,
and $p(\Theta):=(f(U_i,V_j))_{i\in I,j\in J}$ for the prediction map, with Jacobian
$\mathrm{D}p$. For a perturbation $(\Delta U,\Delta V)$, the linearized change in
predictions at $(U^\star,V^\star)$ is
\[
\Delta p^{\lin}_{ij}:=
\langle \nabla_u f(U_i^\star,V_j^\star),\Delta U_i\rangle
+\langle \nabla_v f(U_i^\star,V_j^\star),\Delta V_j\rangle .
\]

\begin{assumption}[Local identifiability up to reparametrization]\label{ass:local-id}
$\mathrm{D}p$ has constant rank near $(U^\star,V^\star)$, and there is $c>0$ such that
\[
\sum_{i,j}\bigl(\Delta p^{\lin}_{ij}\bigr)^2
\ \ge\
c\,h\,\Bigl(\sum_{i}\|\Delta U_i\|^2+\sum_{j}\|\Delta V_j\|^2\Bigr)
\]
for every perturbation $(\Delta U,\Delta V)$ orthogonal to $\ker\mathrm{D}p(U^\star,V^\star)$.
\end{assumption}

The kernel of $\mathrm{D}p$ consists of perturbations that leave every prediction
unchanged to first order; under constant rank, these are exactly the directions
along $\mathcal E_0$, i.e., reparametrizations of the null. \Cref{ass:local-id}
therefore says that item and model perturbations cannot cancel in the predictions
unless they amount to a reparametrization. This is not implied by smoothness
alone, but it holds automatically for IRT/MIRT links.

\begin{lemma}[IRT/MIRT links]\label{lem:irt-id}
Let $f(u,v)=g(u^\top v)$ with $g$ twice continuously differentiable and $g'>0$
(e.g., logistic or probit), and suppose $h>0$. Then \Cref{ass:local-id} holds
with $c=m^2$, where $m:=\min_{i,j}g'\bigl((U_i^\star)^\top V_j^\star\bigr)>0$.
The same holds when some coordinates of every $V_j$ are fixed constants,
e.g., $V_j=(\theta_j,1)$ to encode item intercepts.
\end{lemma}

\begin{proof}
Write $\ell_{ij}:=\langle V_j^\star,\Delta U_i\rangle+\langle U_i^\star,\Delta V_j\rangle$,
so that $\Delta p^{\lin}_{ij}=g'_{ij}\,\ell_{ij}$ with
$g'_{ij}:=g'\bigl((U_i^\star)^\top V_j^\star\bigr)\ge m$. Hence
$\sum_{i,j}(\Delta p^{\lin}_{ij})^2\ge m^2\sum_{i,j}\ell_{ij}^2$, and
$\ker\mathrm{D}p=\ker\ell$.

\emph{Kernel and constant rank.} Stack $U_i^\star$ and $V_j^\star$ as the rows of
$U^\star\in\mathbb R^{|I|\times r}$ and $V^\star\in\mathbb R^{|J|\times r}$, so that
$G_U^\star=(U^\star)^\top U^\star$ and $G_V^\star=(V^\star)^\top V^\star$; since $h>0$,
both matrices have full column rank. If $\ell\equiv 0$, i.e.,
$\Delta U(V^\star)^\top+U^\star\Delta V^\top=0$, multiplying on the right by
$V^\star(G_V^\star)^{-1}$ gives $\Delta U=U^\star B^\top$ for some
$B\in\mathbb R^{r\times r}$, and then $U^\star\bigl(B^\top(V^\star)^\top+\Delta V^\top\bigr)=0$
forces $\Delta V=-V^\star B$. Thus
\[
\ker\mathrm{D}p
=\bigl\{(\Delta U_i,\Delta V_j)=(BU_i^\star,\,-B^\top V_j^\star):\ B\in\mathbb R^{r\times r}\bigr\},
\]
the tangent directions of the reparametrizations $(AU_i,A^{-\top}V_j)$ at $A=I$.
This kernel has dimension $r^2$ at every full-rank $(U,V)$, an open set, so
$\mathrm{D}p$ has constant rank near $(U^\star,V^\star)$.

\emph{Lower bound.} A perturbation is orthogonal to the kernel iff
$\sum_i\langle\Delta U_i,BU_i^\star\rangle=\sum_j\langle\Delta V_j,B^\top V_j^\star\rangle$
for all $B$, i.e., iff
\[
S:=\sum_{i}\Delta U_i(U_i^\star)^\top=\sum_{j}V_j^\star\Delta V_j^\top .
\]
Expanding the square,
\[
\sum_{i,j}\ell_{ij}^2
=\sum_{i,j}\langle V_j^\star,\Delta U_i\rangle^2
+\sum_{i,j}\langle U_i^\star,\Delta V_j\rangle^2
+2\sum_{i,j}\langle V_j^\star,\Delta U_i\rangle\langle U_i^\star,\Delta V_j\rangle .
\]
The first term equals $\sum_i\Delta U_i^\top G_V^\star\Delta U_i\ge\lambda_{\min}(G_V^\star)\sum_i\|\Delta U_i\|^2$,
and similarly the second is at least $\lambda_{\min}(G_U^\star)\sum_j\|\Delta V_j\|^2$.
The cross term is nonnegative:
\[
\sum_{i,j}(V_j^\star)^\top\Delta U_i(U_i^\star)^\top\Delta V_j
=\sum_j(V_j^\star)^\top S\,\Delta V_j
=\operatorname{tr}\Bigl(S\sum_j\Delta V_j(V_j^\star)^\top\Bigr)
=\operatorname{tr}(SS^\top)=\|S\|_F^2\ \ge\ 0 .
\]
Therefore $\sum_{i,j}\ell_{ij}^2\ge h\bigl(\sum_i\|\Delta U_i\|^2+\sum_j\|\Delta V_j\|^2\bigr)$,
and \Cref{ass:local-id} holds with $c=m^2$.

\emph{Fixed coordinates.} If coordinate $k$ of every $V_j$ is fixed, perturbations
have $(\Delta V_j)_k=0$, and kernel directions within this subspace are those with
$B$'s $k$-th column equal to zero. Orthogonality to these gives
$S_U:=\sum_i\Delta U_i(U_i^\star)^\top$ and $S_V:=\sum_j V_j^\star\Delta V_j^\top$
equal outside column $k$, while $S_V$ vanishes on column $k$. The cross term is then
$\langle S_U,S_V\rangle_F=\|S_V\|_F^2\ge0$, and the rest of the argument is unchanged.
\end{proof}

\paragraph{Proof of Theorem~\ref{thm:heterogeneity}.}
\begin{proof}
Note that $\Theta\in\mathcal E_\varepsilon$ iff $\Theta\in\mathcal N$ and
$\|p(\Theta)-p^\star\|\le\varepsilon$, and $\mathcal E_0=p^{-1}(p^\star)$.
Let $k$ be the constant rank of $\mathrm{D}p$ near $\Theta^\star:=(U^\star,V^\star)$ and
$\sigma_k(\cdot)$ its $k$-th singular value, i.e., the smallest nonzero one.
\Cref{ass:local-id} states that $\sigma_k\bigl(\mathrm{D}p(\Theta^\star)\bigr)^2\ge c\,h$.

\paragraph{Step 0: Choose the neighborhood.}
Fix $\rho_0>0$ such that, on the ball $B(\Theta^\star,\rho_0)$:
(a) by the constant rank theorem, $\mathcal E_0\cap B(\Theta^\star,\rho_0)$ is a
$C^1$ submanifold whose tangent space at each point $\bar\Theta$ is
$\ker\mathrm{D}p(\bar\Theta)$;
(b) by continuity of $\sigma_k$ (the rank being constant),
$\sigma_k\bigl(\mathrm{D}p(\bar\Theta)\bigr)^2\ge c\,h/2$;
(c) since $f$ is $C^2$, there is $L<\infty$ with
$\|p(\Theta)-p(\bar\Theta)-\mathrm{D}p(\bar\Theta)(\Theta-\bar\Theta)\|\le L\|\Theta-\bar\Theta\|^2$.
Take $\mathcal N:=B(\Theta^\star,\rho)$ with $\rho:=\min\bigl\{\rho_0/2,\ \sqrt{ch/2}\,/(2L)\bigr\}$.

\paragraph{Step 1: Project onto the nearest exact representation.}
Fix $\Theta\in\mathcal E_\varepsilon$. Since $\mathcal E_0$ is closed and contains
$\Theta^\star$, it has a point $\bar\Theta$ nearest to $\Theta$, and
$\Delta:=\Theta-\bar\Theta$ satisfies
$\|\Delta\|=\operatorname{dist}(\Theta,\mathcal E_0)\le\|\Theta-\Theta^\star\|<\rho$.
Hence $\|\bar\Theta-\Theta^\star\|<2\rho\le\rho_0$, so (a)--(c) apply at $\bar\Theta$
(and on the segment from $\bar\Theta$ to $\Theta$). As $\bar\Theta$ minimizes the
distance to $\Theta$ over the manifold in (a), $\Delta$ is orthogonal to its
tangent space: $\Delta\perp\ker\mathrm{D}p(\bar\Theta)$.

\paragraph{Step 2: Lower bound the change in predictions.}
Since $p(\bar\Theta)=p^\star$, (c) and (b) give
\[
\varepsilon
\ \ge\ \|p(\Theta)-p^\star\|
\ \ge\ \|\mathrm{D}p(\bar\Theta)\Delta\|-L\|\Delta\|^2
\ \ge\ \Bigl(\sqrt{c h/2}-L\|\Delta\|\Bigr)\|\Delta\|
\ \ge\ \tfrac12\sqrt{c h/2}\,\|\Delta\|,
\]
where the last step uses $\|\Delta\|<\rho\le\sqrt{ch/2}\,/(2L)$.

\paragraph{Step 3: Conclude.}
Rearranging, $\operatorname{dist}(\Theta,\mathcal E_0)=\|\Delta\|\le\sqrt{8/c}\;\varepsilon/\sqrt h$.
As $\Theta\in\mathcal E_\varepsilon$ was arbitrary, the theorem holds with
$C=\sqrt{8/c}$ (so $C=\sqrt8/m$ for IRT/MIRT links, by \Cref{lem:irt-id}).
\end{proof}

\section{Experiment Details}
\subsection{Inferring IRT parameters}\label{appsec:infer_irt}
We estimate the IRT parameters $\{a_i, d_i\}_{i=1}^n$ and $\{\theta_j\}_{j=1}^m$ via gradient-ascent on the joint log-likelihood

\[
\ell(a,d,\theta)
\;=\;
\sum_{i=1}^n\sum_{j=1}^m 
\Bigl[ 
Y_{ij}\log p_{ij} + (1-Y_{ij})\log(1-p_{ij})
\Bigr],
\]
with $\ell_2$ penalties on $a_i$, $b_i$, and~$\theta_j$ to ensure stability and identifiability. We also apply periodic whitening transformations to $\theta_j$ so that the covariance of $\theta_j$ remains close to the identity; this removes rotational and scale indeterminacies inherent to IRT models, see \cite{reckase2006multi} for details. We then recover the estimated probability of correctness (where $\Phi$ is the cdf of a standard normal): 
\[
\hat p_{ij} = \Phi\!\left( \hat a_i^{\top}\hat \theta_j + \hat d_i \right).
\]

\subsection{Experiment 1 (Details)}\label{appsec:details_exp2}

\paragraph{Baselines from prior work.}
Both \citet{kimcorrelated2025} and \citet{goelgreat2025} operate on multiple-choice data.
Let $A_{ij}\in\{1,\dots,K_i\}$ denote the option selected by model $j$ on item $i$, with correct option $a_i^\star$.

\emph{Kim et al.} focus on agreement conditional on both models being incorrect:
\[
q_{j\ell}:=\Pr(A_{ij}=A_{i\ell}\mid A_{ij}\neq a_i^\star,\;A_{i\ell}\neq a_i^\star),
\]
and compare this quantity to the chance baseline
\[
q_{j\ell}^{\text{ch}} := \mathbb E\!\left[\frac{1}{K_i-1}\right],
\]
which corresponds to uniform random selection among the incorrect options.
Excess correlation is defined as the difference between these two quantities:
$$\text{Excess} := q_{j\ell} - q_{j\ell}^{\text{ch}}$$

$$\kappa_{j\ell}^{\text{err}} := \frac{q_{j\ell} - q_{j\ell}^{\text{ch}}}{1- q_{j\ell}^{\text{ch}}}$$

\emph{Goel et al.} define a probabilistic agreement measure based on the inner product of predictive distributions.
In the discrete (argmax-only) setting, their observed agreement reduces to
\[
c_{\mathrm{obs}}
\;=\;
\frac{1}{n}\sum_{i=1}^n \mathbf 1\{A_{ij}=A_{i\ell}\},
\]
while their chance agreement baseline is
\[
c_{\mathrm{exp}}
\;=\;
p_j p_\ell
\;+\;
(1-p_j)(1-p_\ell)\,\mathbb E\!\left[\frac{1}{K_i-1}\right],
\]
where $p_j$ denotes model $j$'s overall accuracy.
Their CAPA statistic is a normalized excess-over-chance measure,
\[
\mathrm{CAPA}_{j\ell}
\;=\;
\frac{c_{\mathrm{obs}}-c_{\mathrm{exp}}}{1-c_{\mathrm{exp}}}.
\]

Both baselines adjust for differences in model capability, but treat all items as exchangeable.

\paragraph{Our null: 1-dimensional IRT on multiple-choice correctness.}
We now evaluate excess correlation using a strictly more expressive null model.
Specifically, we fit a \emph{1-dimensional IRT model} ($K=1$) to the same multiple-choice data, where
\[
Y_{ij} := \mathbf 1\{A_{ij}=a_i^\star\}
\]
indicates whether model $j$ selected the correct option on item $i$.
The null model is
\[
\Pr(Y_{ij}=1)
\;=\;
\Phi(a_i \theta_j + b_i),
\]
with item-specific difficulty and discrimination parameters $(a_i,b_i)$ and model ability $\theta_j$.

This model strictly subsumes the baselines above: if all items share identical parameters, the IRT model reduces to a capability-only model, whereas allowing $(a_i,b_i)$ to vary introduces item heterogeneity.
Thus, in the language of Section~\ref{sec:absorption}, this represents a higher rung in the null ladder.

\paragraph{Residual dependence.}
After fitting the model, we compute residuals
\[
R_{ij}=Y_{ij}-\hat p_{ij},
\]
de-bias them as in Experiment~1, and estimate the residual covariance matrix
\[
\widehat C
\;=\;
\frac{1}{n}\tilde R^\top \tilde R.
\]
This quantity captures excess model-to-model dependence not explained by either model ability or item difficulty.

\subsubsection{Additional Results}\label{appsec:addl_results_previous_papers}
In Figure~\ref{fig:exp2_heatmaps}, we used a null model where each item is a (question, answer-choice) pair in order to provide a fair comparison to previous methods. In Figure~\ref{fig:exp2_heatmaps_addl}, we use a null model where each item is at the question level. The results do not change significantly, indicating that our findings remain robust to the choice of null model.

\begin{figure*}
    \centering
    \includegraphics[width=\linewidth]{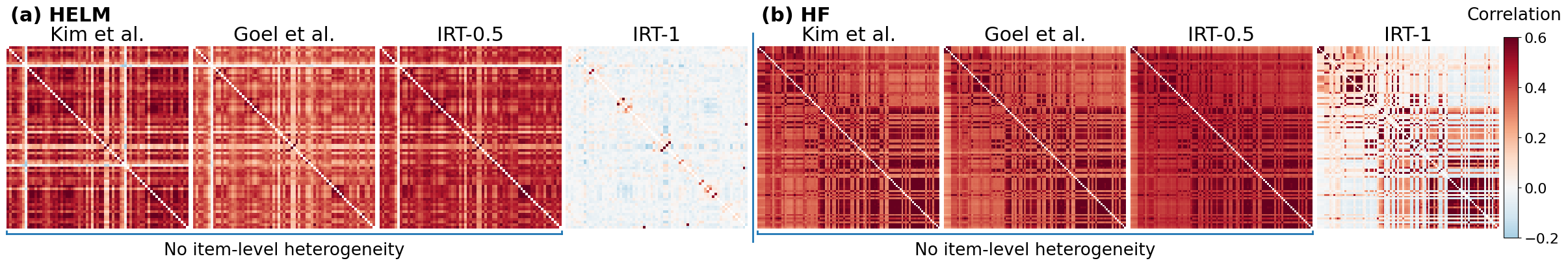}
    \caption{Residual correlation matrices for models in \textbf{HELM} (a) and \textbf{HF} (b), using different nulls: from left to right, a baseline from \citet{kimcorrelated2025}, from \citet{goelgreat2025}, from a 1D IRT with no item difficulty, and from a 1D IRT with item difficulties. Each item is a question. The first three null models \textbf{do not} include item heterogeneity. As such, the corresponding excess correlation for the full IRT model is attenuated compared to the others because item difficulties absorb much of the apparent positive correlation.}
\label{fig:exp2_heatmaps_addl}
\end{figure*}

\subsection{Figure~\ref{fig:off_diag_main}: Increasing dimensions reduces excess correlation}\label{appsec:irt_k_exp}

\textbf{Setup.} Consider a sequence of IRT models with increasing dimensions $K$. For each $K\in\{1, 2, 4, 8, 16, 32, 36, \dots, 60, 64\}$, we fit a $K$-dimensional ability model via the procedure from Section~\ref{sec:comparison_previous}.
For each $K$, we report fit via mean squared error:
$\mathrm{MSE}^{(K)}\ :=\ \frac{1}{nm}\sum_{i=1}^n\sum_{j=1}^m\big(Y_{ij}-\hat p_{ij}^{(K)}\big)^2.$ We also calculate the excess correlation matrix $\widehat{\Sigma}$ for each $K$. 

In order to summarize the pairwise covariance distribution, we compute the following summary statistics from $C$ for each $K$:
\[
\text{(i) AbsMax}^{(K)}:=\max_{j\neq \ell}\big|\widehat \Sigma^{(K)}_{j\ell}\big|,\qquad
\text{(ii) AbsMean}^{(K)}:=\frac{1}{m(m-1)}\sum_{j\neq \ell}\big|\widehat \Sigma^{(K)}_{j\ell}\big|,
\]
\[
\text{(iii) Median}^{(K)}:=\operatorname{median}_{j\neq \ell}\big|\widehat \Sigma^{(K)}_{j\ell}\big|,\qquad
\text{(iv)}\;\|\widehat \Sigma^{(K)}\|_{\mathrm{off},F}:=\sqrt{\sum_{j\neq \ell}\big(\widehat \Sigma^{(K)}_{j\ell}\big)^2}.
\]

\paragraph{Results.}
Figure~\ref{fig:off_diag} shows the quantities of interest ($y$-axis) across $K$ dimensions of the IRT model ($x$-axis) on both \textbf{HELM} and \textbf{HF}. By Proposition~\ref{prop:monotone_absorption}, enlarging the feasible class from $\mathcal N_K$ to $\mathcal N_{K+1}$ cannot increase the optimal discrepancy, and we indeed observe a monotone decrease of $\mathrm{MSE}^{(K)}$ as $K$ increases. Similarly, by Theorem~\ref{thm:residual-dependence-vanishes}, we see that the pairwise residual covariances go to 0 as $K$ increases. This result suggests that unexplained correlation vanishes for sufficiently complex null models. Note that the figures are on a log scale, meaning that increasing dimensions exponentially decreases excess correlation.

\begin{figure}
    \centering
    \includegraphics[width=\linewidth]{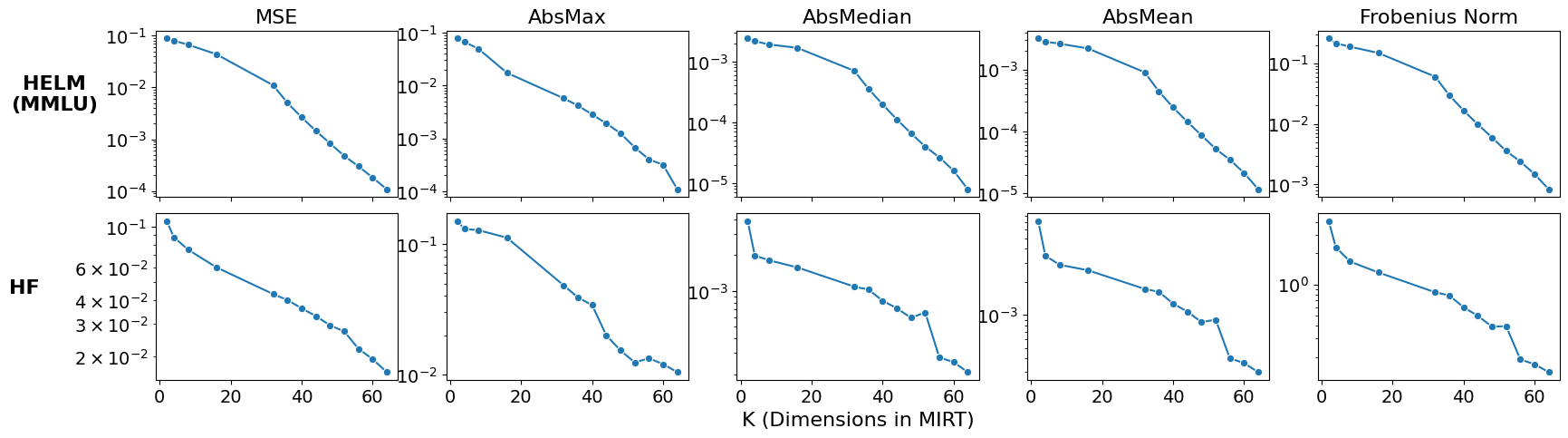}
    \caption{Mean square error (MSE) and summary statistics of the distribution of pairwise residual correlations, as a function of $K$: the dimensions in the multidimensional IRT model. Top (bottom) row shows results on the \textbf{HELM} (\textbf{HF}) dataset. As $K$ increases, the residual correlations that are unexplained by the independent null model tend toward zero, meaning that a sufficiently expressive null model can arbitrarily absorb model correlations.}
    \label{fig:off_diag}
\end{figure}


\section{Extensions}
\subsection{Generalizing to other output spaces}\label{appsec:generalizations}
The framework accommodates a wide range of choices for $\mathcal Y$, including:
\begin{itemize}[leftmargin=2em]
\item[(1)] \textbf{Binary correctness:} $\mathcal Y = \{0,1\}$, where $Y_{ij}=1$ denotes a correct answer.
\item[(2)] \textbf{Probabilistic correctness:} $\mathcal Y = [0,1]$, where $Y_{ij}$ is a calibrated probability of correctness.
\item[(3)] \textbf{Logits or score vectors:} $\mathcal Y = \mathbb R^V$, such as the model’s pre-softmax logits over a vocabulary of size $V$.
\item[(4)] \textbf{Natural language:} $\mathcal Y = \mathcal T^\ast$, the space of token sequences.
\end{itemize}

Generalizing to probabilistic correctness is immediate: one simply replaces each Bernoulli
variable with a real value in $[0,1]$ and defines independence conditional on latent
parameters in the same manner. Richer output spaces (logits and natural language) require more care. For logits, one natural approach is to view each $Y_{ij}$ as a vector of real-valued features and study conditional independence of these vectors given latent parameters. For natural language, a possible strategy is to embed outputs via a mapping $f:\mathcal T^\ast \to \mathbb R^d$ (e.g.\ semantic or style embeddings) and treat the embedded representations $X_{ij}=f(Y_{ij})$ as the objects of analysis. In this case, an ``independent'' null model would posit that, conditional on latent structure, these embeddings are independent across models. However, the appropriate parameterization of such latent structure (e.g.\ global factors capturing topic, stylistic preferences, or discourse-level strategies) is not yet clear and presents an interesting direction for future work.

\subsection{Inferring $\Sigma$}\label{appsec:infer_sigma_details}

In this work, we define and calculate monoculture $\Sigma$ as the pairwise correlations of the residuals between the outcomes $Y$ and the null model fitted probabilities $\hat{p}$. This gives us a metric that is agnostic to the particular inference algorithm used to infer $\Sigma$. However, one may want to directly infer $\Sigma$ from the data, and in this section we provide one principled way of doing so.

\paragraph{Kernel parameterization and estimation of $\Sigma$.}

We first introduce a latent Gaussian variable $Z_{ij}$, where $Y_{ij} = \mathbb{I}[Z_{ij} \geq 0]$. We can therefore equivalently re-write the generative model as:

\begin{equation}
Z_{ij} \;=\; \theta_j - b_i + \varepsilon_{ij},
\qquad 
\varepsilon_{i,\cdot} \sim \mathcal{N}(0,\Sigma),
\label{eq:1d-irt-alt}
\end{equation}
where for each item $i$ the vector $\varepsilon_{i\cdot} \in \mathbb{R}^M$ (across $M$ models) has covariance matrix $\Sigma \in \mathbb{R}^{M \times M}$, shared across items. 
The null model corresponds to $\Sigma = I_M$, i.e., conditionally independent models given $(\theta,b)$. Any deviation $\Sigma \neq I_M$ is interpreted as \emph{excess correlation} between models beyond what is induced by sharing the same item difficulties and model abilities.

In the alternative model~\eqref{eq:1d-irt-alt}, assume the excess-correlation matrix
$\Sigma \in \mathbb{R}^{M \times M}$ is parameterized via low-dimensional
embeddings for each model. Let $v_j \in \mathbb{R}^K$ denote the embedding for
model $j$, and let $V \in \mathbb{R}^{M \times K}$ collect these row-wise.
We define pairwise correlations by
\[
\rho_{jk} \;=\; v_j^\top v_k, \qquad j \neq k,
\]
and enforce unit variance by setting
\begin{equation}
\Sigma \;=\; VV^\top \;+\; \operatorname{diag}\!\bigl(1 - \|v_1\|_2^2,\dots,1 - \|v_M\|_2^2\bigr).
\label{eq:sigma-kernel}
\end{equation}
If we maintain $\|v_j\|_2 \le 1$ for all $j$, then $\Sigma$ is positive
semidefinite and has unit diagonal:
\[
\Sigma_{jj} \;=\; \|v_j\|_2^2 + (1 - \|v_j\|_2^2) \;=\; 1.
\]
In the implementation, we explicitly use
$\rho_{jk} = v_j^\top v_k$ for $j\neq k$ in the bivariate probit calculations
and fix $\Sigma_{jj} = 1$ via~\eqref{eq:sigma-kernel}.

\medskip
\noindent\textbf{Model-implied and empirical covariances.}
Let $Y \in \{0,1\}^{N \times M}$ be the correctness matrix (items $\times$ models),
and let $(\hat{\theta},\hat{b})$ be the 1D IRT estimates from Stage~1. Define
\[
\hat{\mu}_{ij} = \hat{\theta}_j - \hat{b}_i,
\qquad
\hat{p}_{ij} = \Phi(\hat{\mu}_{ij}),
\qquad
\hat{m}_j = \frac{1}{N}\sum_{i=1}^N \hat{p}_{ij},
\qquad
\bar{Y}_j = \frac{1}{N}\sum_{i=1}^N Y_{ij}.
\]

For a given pair $(j,k)$ and correlation $\rho_{jk}$, the latent pair
$(Z_{ij},Z_{ik})$ for item $i$ is modeled as
\[
\begin{pmatrix} Z_{ij} \\[2pt] Z_{ik} \end{pmatrix}
\sim
\mathcal{N}
\!\left(
\begin{pmatrix} \hat{\mu}_{ij} \\[2pt] \hat{\mu}_{ik} \end{pmatrix},
\begin{pmatrix}
1 & \rho_{jk} \\
\rho_{jk} & 1
\end{pmatrix}
\right).
\]
The model-implied joint success probability for item $i$ is
\[
p_{11}^{(i)}(j,k;\rho_{jk})
\;:=\;
\mathbb{P}\bigl(Y_{ij}=1, Y_{ik}=1\bigr)
\;=\;
\Phi_2\bigl(\hat{\mu}_{ij}, \hat{\mu}_{ik}; \rho_{jk}\bigr),
\]
where $\Phi_2(\cdot,\cdot;\rho)$ is the bivariate normal CDF with correlation
$\rho$. The model-implied covariance between $Y_{\cdot j}$ and $Y_{\cdot k}$ is
\begin{equation}
C^{\text{model}}_{jk}(V)
\;=\;
\frac{1}{N} \sum_{i=1}^N
p_{11}^{(i)}(j,k;\rho_{jk})
\;-\;
\hat{m}_j \hat{m}_k,
\qquad
\rho_{jk} = v_j^\top v_k.
\label{eq:C-model}
\end{equation}
The empirical covariance is
\begin{equation}
C^{\text{emp}}_{jk}
\;=\;
\frac{1}{N}\sum_{i=1}^N Y_{ij}Y_{ik}
\;-\;
\bar{Y}_j \bar{Y}_k.
\label{eq:C-emp}
\end{equation}
For each iteration, the implementation computes these quantities in a
streaming/tiled fashion over items and model blocks, but conceptually they are
just~\eqref{eq:C-model}–\eqref{eq:C-emp}.

\medskip
\noindent\textbf{Loss and gradients in correlation space.}
We fit $V$ by matching model-implied and empirical covariances off the diagonal.
The loss is
\begin{equation}
\mathcal{L}(V)
\;=\;
\sum_{1 \le j < k \le M}
\bigl(C^{\text{model}}_{jk}(V) - C^{\text{emp}}_{jk}\bigr)^2.
\label{eq:loss-sigma}
\end{equation}
To differentiate with respect to $\rho_{jk}$, note that
\[
\frac{\partial}{\partial \rho_{jk}}
p_{11}^{(i)}(j,k;\rho_{jk})
\;=\;
\varphi_2\bigl(\hat{\mu}_{ij},\hat{\mu}_{ik};\rho_{jk}\bigr),
\]
where $\varphi_2$ is the bivariate normal PDF with correlation $\rho_{jk}$.
Hence
\[
\frac{\partial C^{\text{model}}_{jk}}{\partial \rho_{jk}}
\;=\;
\frac{1}{N}\sum_{i=1}^N
\varphi_2\bigl(\hat{\mu}_{ij},\hat{\mu}_{ik};\rho_{jk}\bigr)
\;=:\;
\bar{\varphi}_{jk},
\]
and by the chain rule,
\begin{equation}
\frac{\partial \mathcal{L}}{\partial \rho_{jk}}
\;=\;
2\bigl(C^{\text{model}}_{jk} - C^{\text{emp}}_{jk}\bigr)
\bar{\varphi}_{jk}.
\label{eq:dL-drho}
\end{equation}
In the code, $C^{\text{model}}_{jk} - C^{\text{emp}}_{jk}$ is stored as
\texttt{diff\_blk}, and $\bar{\varphi}_{jk}$ as \texttt{phibar\_blk}, so that
\[
G_{jk}
\;:=\;
\frac{\partial \mathcal{L}}{\partial \rho_{jk}}
\;=\;
2\,\texttt{diff\_blk}_{jk}\,\texttt{phibar\_blk}_{jk}.
\]

\medskip
\noindent\textbf{Gradients in embedding space and optimization.}
Because $\rho_{jk} = v_j^\top v_k$ for $j\neq k$, we have
\[
\frac{\partial \rho_{jk}}{\partial v_j} = v_k,
\qquad
\frac{\partial \rho_{jk}}{\partial v_k} = v_j,
\]
and thus
\begin{equation}
\frac{\partial \mathcal{L}}{\partial v_j}
\;=\;
\sum_{k\neq j}
\frac{\partial \mathcal{L}}{\partial \rho_{jk}}
\frac{\partial \rho_{jk}}{\partial v_j}
\;=\;
\sum_{k\neq j} G_{jk} \, v_k.
\label{eq:dL-dvj}
\end{equation}
In matrix form, if $G$ is the matrix with entries $G_{jk}$ (symmetric and
zero on the diagonal), then the gradient w.r.t.\ $V$ is
\[
\nabla_V \mathcal{L}
\;=\;
G V,
\]
which is exactly what the code accumulates in the tensor \texttt{Z}.
We then perform Adam updates on $V$:
\[
V \leftarrow \operatorname{AdamStep}\bigl(V,\nabla_V \mathcal{L}\bigr),
\]
with bias-corrected learning rate as in standard Adam.

\medskip
\noindent\textbf{Norm projection and unit variance.}
After each Adam update, we project each row $v_j$ onto the unit ball:
\begin{equation}
v_j \;\leftarrow\; 
\frac{v_j}{\max\{1,\|v_j\|_2\}}.
\label{eq:projection}
\end{equation}
This ensures $\|v_j\|_2 \le 1$ for all $j$, so that the covariance matrix
$\Sigma$ in~\eqref{eq:sigma-kernel} remains positive semidefinite with unit
diagonal. Conceptually, we are learning a low-rank correlation structure
$\rho_{jk} = v_j^\top v_k$ subject to the constraint that all variances are
fixed at one, which preserves the scaling of the underlying 1D IRT model.

\medskip
\noindent\textbf{Practical details.}
The implementation evaluates the loss~\eqref{eq:loss-sigma} and its gradients
exactly but does so in a memory-efficient way by:
(i) streaming over items in chunks (size \texttt{item\_chunk}) and
(ii) tiling the model index set into blocks (size \texttt{model\_block}) to
compute $C^{\text{model}}_{jk}$ and $C^{\text{emp}}_{jk}$ blockwise.
We use early stopping based on the best observed loss over iterations.

\textbf{Justification for two-stage estimation.}
This two-stage procedure separates \emph{first-order} and \emph{second-order} structure: Stage~1 uses only the marginal success probabilities to learn $(\theta,b)$, while Stage~2 uses the residual covariances to learn $\Sigma$. Under standard regularity conditions, the Stage~1 maximum-likelihood estimates $(\hat{\theta},\hat{b})$ are consistent and $\sqrt{N}$-accurate even if $\Sigma \neq I_M$, because the marginal Bernoulli probabilities depend only on $(\theta,b)$ and not on $\Sigma$. Plugging these into the second-stage moment-matching objective yields a consistent estimator of $\Sigma$ that is asymptotically equivalent to a joint estimator but computationally simpler. Conceptually, this is analogous to generalized method-of-moments procedures that estimate parameters associated with lower-order moments before fitting higher-order structure.



\end{document}